\documentclass[journal]{IEEEtran}

\usepackage{amsmath}
\usepackage{amsfonts}
\usepackage{amssymb}
\usepackage{amsthm}
\usepackage{tabularx}
\usepackage{mathrsfs} 
\usepackage{mathtools}

\usepackage{url}
\usepackage{hyperref}
\newtheorem{theorem}{Theorem}

\newtheorem{corollary}{Corollary}

\usepackage{stfloats}
\usepackage{float}
\usepackage{graphicx}
\usepackage{xcolor}
\usepackage{subfigure}
\renewcommand{\vec}[1]{\mathbf{#1}}

\makeatletter
\def\blfootnote{\xdef\@thefnmark{}\@footnotetext}
\makeatother

\begin{document}

\title{\LARGE{Cache-Enabled Millimetre-Wave Fluid Antenna Systems:\\Modeling and Performance}} 
\author{Farshad~Rostami~Ghadi,~\IEEEmembership{Member},~\textit{IEEE}, 
            Kai-Kit~Wong,~\IEEEmembership{Fellow},~\textit{IEEE},\\
            Kin-Fai~Tong,~\IEEEmembership{Fellow},~\textit{IEEE},~and 
            Yangyang Zhang
\vspace{-1cm}
\thanks{The work of F. Rostami Ghadi, K. K.  Wong and K. F. Tong is supported by the Engineering and Physical Sciences Research Council (EPSRC) under Grant EP/W026813/1.}
\thanks{F. Rostami Ghadi, K. K. Wong and K. F. Tong are with the Department of Electronic and Electrical Engineering, University College London, London, WC1E 6BT, United Kingdom. K. K. Wong is also with Yonsei Frontier Lab, Yonsei University, Seoul, 03722, Korea (e-mail: $\{\rm f.rostamighadi,\rm kai\text{-}kit.wong,k.tong\}@ucl.ac.uk$).}
\thanks{Y. Zhang is with Kuang-Chi Science Limited, Hong Kong SAR, China.}
}
\maketitle

\begin{abstract}
This letter investigates the performance of content caching in a heterogeneous cellular network (HetNet) consisting of fluid antenna system (FAS)-equipped mobile users (MUs) and millimeter-wave (mm-wave) single-antenna small base stations (SBSs), distributed according to the independent homogeneous Poisson point processes (HPPP). In particular, it is assumed that the most popular contents are cached in the SBSs to serve the FAS-equipped MUs requests. To assess the system performance, we derive compact expressions for the successful content delivery probability (SCDP) and the content delivery delay (CDD) using the Gauss-Laguerre quadrature technique. Our numerical results show that the performance of cache-enabled mm-wave HetNets can be greatly improved, when the FAS is utilized at the MUs instead of traditional fixed-antenna system deployment.
\end{abstract}


\blfootnote{Digital Object Identifier 10.1109/XXX.2021.XXXXXXX.}
\blfootnote{Corresponding author: F. Rostami Ghadi.}

\section{Introduction}\label{sec-intro}
\IEEEPARstart{G}{lobal} mobile data traffic has experienced an explosive growth over recent years, due to rapid development in mobile technologies and an increasing demand for using smart devices with higher data rates. According to the Ericsson report, by $2028$, the monthly global mobile data traffic is projected to soar to $325$ exabytes (EB), which is almost quadruple compared to the levels recorded in $2022$ \cite{CiteDrive2022}. In this regard, video traffic, which is mainly due to a tremendous demand for video streaming in multimedia applications, accounts for most of the data traffic. One key solution to tackle this issue in future mobile networks, i.e., sixth-generation (6G) technology, is to deploy dense small-cell base stations (SBSs) in the existing macrocell cellular networks, so that the most popular contents are cached at the SBSs in advance during off-peak hours \cite{bastug2014living}. By doing so, contents become much closer to mobile users (MUs), and once a content (e.g., video) is requested by a MU, the corresponding SBS can directly serve the MU and thus the end-to-end content delivery delay (CDD) reduces and the MU's quality of service (QoS) improves.

Independently, the main challenge at the MUs is related to their physical space limitations, which dictates the advantage of using multi-antenna systems in providing capacity gains by exploiting diversity over multiple signals affected by fading. More precisely, the rule of thumb in multiple-input multiple-output (MIMO) systems is to spatially separate the antennas by at least half of a wavelength to ensure diversity gain, making it impractical to employ massive MIMO in a tiny space of mobile devices. Fluid antenna system (FAS) has recently been introduced as a promising alternative to provide significant diversity gain in the small space of mobile devices via antenna position flexibility. More specifically, FAS refers to any software-manageable fluidic, conductive, or dielectric structure capable of changing its shape and size, and switching its position (i.e., ports) over a pre-defined space to reconfigure the radiation characteristics \cite{wong2020fluid}. This unique feature can provide enhanced performance, flexibility, and durability for MUs in various wireless communication systems.

In recent years, extensive contributions have been carried out for cache-enabled heterogeneous cellular networks (HetNets) from various aspects, such as cache placement strategies \cite{blaszczyszyn2015optimal,wang2017cooperative,zhu2018content} and content delivery techniques \cite{bacstug2015cache,rostami2019outage,wang2021content}. On the other hand, great efforts have been made regarding FAS and its application to enhance wireless system performance, focusing on channel modeling and estimation \cite{wong2021fluid,khammassi2023new,ghadi2023gaussian,skouroumounis2022fluid,xu2023channel,zhang2023successive}, and performance analysis \cite{tlebaldiyeva2022enhancing,new2023fluid,ghadi2023copula,vega2023novel,ghadi2023fluid,ghadi2024fluid}, and optimization \cite{10354059,ye2023fluid,new2023fluid1,ISAC_FAS}. Nevertheless, the dynamic and reconfigurable properties of FAS are beneficial to content-centric cellular networks but this is not well understood.

The remarkable feature of cache-enabled cellular networks in bringing content closer to MUs when combined with the unique feature of FAS in adaptively modifying its radiating structures based on network demands can potentially be synergistic. Motivated by this, this letter's aim is to evaluate the performance of cache-enabled cellular networks when MUs take advantage of FAS. In particular, we first consider a dense cellular network, where the most popular content is cached in the millimetre-wave (mm-wave) SBSs to serve the FAS-equipped MUs. Then, in order to analyze the considered system performance, we derive compact analytical expressions for the successful content deliver probability (SCDP) and CDD in terms of the joint multivariate normal distribution, exploiting the Gauss-Laguerre technique. Consequently, numerical results indicate that deploying FAS with only one activated port at the MUs can remarkably enhance the system performance rather than using the traditional fixed-antenna system at the MUs.

\section{System Model}\label{sec-sys}
\subsection{Caching Strategy and Content Placement}
We consider a wireless cache-enabled HetNet consisting of mm-wave SBSs tiers and MUs as depicted in Fig.~\ref{fig-model}, where the SBSs have limited storage that caches the most popular contents. In this network, it is assumed that each MU in the corresponding tier connects to its nearest SBS that has cached the desired content. For this purpose, we suppose a finite set of content, denoted as $\mathcal{F}\coloneqq \left\{f_1,\dots,f_l,\dots, f_L\right\}$, in which $f_l$ is the $l$-th most popular content and $L$ represents the number of contents. Moreover, it is assumed that each SBS is able to cache a maximum of $K$ contents so that $K\ll L$, where the size of each content is normalized to $1$. Furthermore, we define the request probability for the $l$-th content as $p_l$ which is modeled as the Zipf distribution, i.e.,
\begin{align}
p_l = \frac{l^{-\zeta}}{\sum_{k=1}^{L}k^{-\zeta}}, 
\end{align}
where $\zeta$ is the Zipf exponent and $\sum_{l=1}^{L}p_l=1$. Additionally, without loss of generality, we proceed under the assumption that the contents are arranged in descending order based on the values of $p_l$. We also consider a probabilistic caching model in which content is stored independently with equal probability across all SBSs within the same tier. Under such an assumption, we denote $q_l$ as the probability that the $l$-th content is cached at a SBS. Therefore, in this probabilistic caching strategy, the set of caching probability $\textit{q}=\left\{q_1,\dots,q_l,\dots,q_L\right\}$ requires to satisfy the conditions $\sum_{l=1}^{L} b_l\leq K$.

\begin{figure}[]
\centering
\includegraphics[width=1\columnwidth]{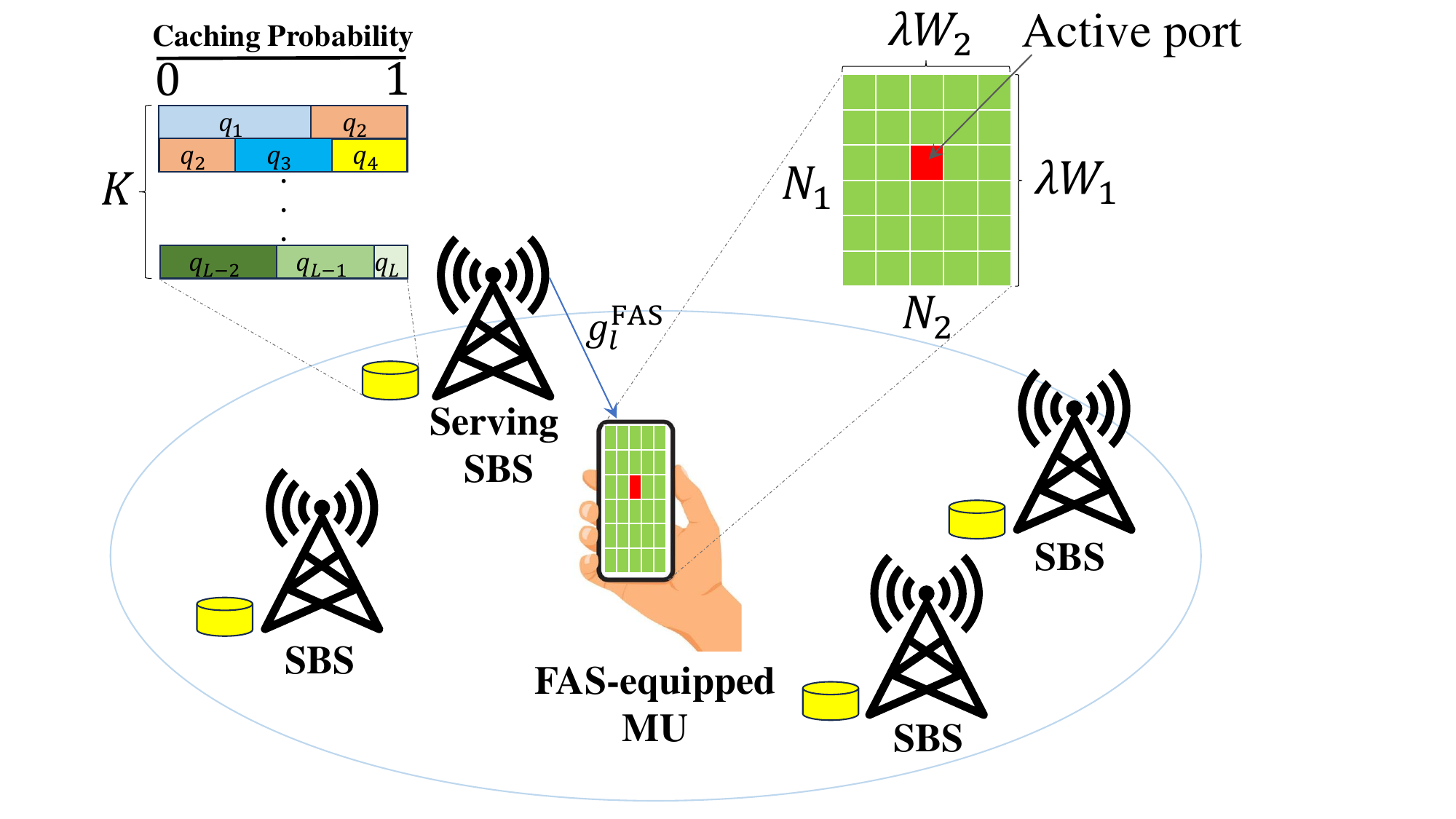}\vspace{-0.8cm}
\caption{A cache-enabled mm-wave HetNet with FAS-equipped MUs.}\vspace{-0.5cm}
\label{fig-model}
\end{figure}

\subsection{Signal Model and Downlink Transmission}
We consider that the locations of the SBSs are distributed according to an independent homogeneous Poisson point process (HPPP), $\Psi^\mathrm{S}$ with intensity $\mu_\mathrm{S}$. Specifically, $\Psi^\mathrm{S}_l$ denotes the point process corresponding to all the SBSs that cache the content $l$ in the mm-wave tier with intensity $q_l\mu_\mathrm{S}$. We assume that a  typical FAS-equipped MU is located at the origin of a Cartesian plane, i.e., $\mathrm{o}=\left(0,0\right)$. Each SBS is equipped with a traditional fixed-antenna, but each MU has a planar FAS. The intensity of MUs is typically much higher than the intensity of SBSs, and we assume that a FAS MU is only permitted to communicate with the SBS within a given time slot.\footnote{This assumption can be managed by employing multiple access techniques in dense small cell networks \cite{zhu2016wireless}.}

Additionally, we consider a grid structure for the respective FAS so that $N_i$ ports are uniformly distributed along a linear space of length $W_i\lambda$ for $i\in\{1,2\}$, i.e., $N=N_1\times N_2$ and $W=W_1\lambda\times W_2\lambda$, where $\lambda$ denotes the wavelength
associated with the carrier frequency. We also define $\mathcal{G}\left(n\right)=\left(n_1,n_2\right)$ and $n=\mathcal{G}^{-1}\left(n_1,n_2\right)$ as an appropriate mapping function to convert the two-dimensional (2D) indices to the one-dimensional (1D) index, in which $n\in\left\{1,\dots,N\right\}$ and $n_i\in\left\{1,\dots,N_i\right\}$. As a result, under such assumptions, the received signal-to-noise ratio (SNR) at a MU located at the origin that requests the content $l$ from the corresponding SBS $X_\mathrm{o}$ that includes this content is defined as
\begin{align}\label{eq:snr}
\gamma_l=\frac{P\left| h_l^{n^*}\right|^2 D\left(\vert X_l\vert\right)}{\sigma^2},
\end{align}
where $P$ denotes the transmit power and $D\left(\vert X_l\vert\right)=\beta\vert X_l\vert^{-\alpha}$ defines the pathloss with the distance $\vert X_l\vert$, with $\beta$ and $\alpha$ being the frequency dependent constant parameter and the pathloss exponent, respectively. Recognizing the fact that mm-wave communications tend to be noise-limited and interference is relatively minimal, $\sigma^2$ in (\ref{eq:snr}) denotes the combined power of noise and weak interference at the MU.\footnote{In dense networks, mm-wave communication operates in the noise-limited regime since the significant path loss diminishes interference and potentially enhances signal directivity \cite{singh2015tractable}. In fact, unlike the sub-$6$ GHz bands, which are typically interference-limited, mm-wave networks primarily face noise-limited, especially when the base station density is not exceedingly high. This phenomenon is mainly due to the narrow beam and blocking effects \cite{andrews2016modeling}.} Moreover, $n^*$ is the index of the best port, which is given by
\begin{align}
n^* = \arg\max_n\left\{\left| h_l^{n}\right|^2\right\},
\end{align}
where $h_l^n$ is the Rayleigh fading channel coefficient between the $n$-th port of a typical FAS-equipped MU and its serving SBS so that the channel gain $g_l^n=\left| h_l^{n}\right|^2$ follows exponential distribution, i.e., $g_l^n\sim\mathrm{Exp}\left(1\right)$. Therefore, the equivalent channel gain at the MU is given by
\begin{align}
g_l^\mathrm{FAS} = \max\left\{g_l^1,\dots,g_l^N\right\}.
\end{align}

In addition, regarding the fact that the FAS ports can freely switch and be arbitrarily close to each other, the channel coefficients are spatially correlated. Thus, considering a three-dimensional (3D) environment under rich scattering, the covariance between two arbitrary ports $n=\mathcal{G}^{-1}\left(n_1,n_2\right)$ and $\tilde{n}=\mathcal{G}^{-1}\left(\tilde{n}_1,\tilde{n}_2\right)$ is characterized as
\begin{align}
\hspace{-3mm}\varrho_{n,\tilde{n}}=j_0\left(2\pi\sqrt{\left(\frac{n_1-\tilde{n}_1}{N_1-1}W_1\right)^2+\left(\frac{n_2-\tilde{n}_2}{N_2-1}W_2\right)^2}\right),
\end{align}
where $\tilde{n}\in\left\{1,\dots,N\right\}$ and $\tilde{n}_i\in\left\{1,\dots,N_i\right\}$, and $j_0(\cdot)$ is the zeroth-order spherical Bessel function of the first kind.

Given that only the optimal port that maximizes the received SNR at the FAS-equipped MU is activated, the cumulative distribution function (CDF) and probability density function (PDF) of $g_l^\mathrm{FAS}$ are, respectively, given by \cite{ghadi2023gaussian}
\begin{align}\label{eq-cdf-fas}
F_{g^\mathrm{FAS}_l}\left(r\right)=\Phi_\vec{R}\left(\phi^{-1}\left(F_{g_l^1}\left(r\right)\right),\dots,\phi^{-1}\left(F_{g_l^N}\left(r\right)\right);\vartheta\right),
\end{align}
where $\Phi_\vec{R}$ represents the joint CDF of the multivariate normal distribution with zero mean vector and correlation matrix $\vec{R}$, $\phi^{-1}\left(F_{g_l^n}\left(r\right)\right)=\sqrt 2\mathrm{erf}^{-1}\left(2F_{g_l^n}\left(r\right)-1\right)$ represents the quantile function of the standard normal distribution, and $\vartheta$ is the dependence parameter, i.e., the Gaussian copula parameter \cite{ghadi2023gaussian}. Besides, since all channels undergo Rayleigh fading, $g_l^n$ for $n\in\left\{1,\dots,N\right\}$ has the following distributions:
\begin{align} \label{eq-dist}
F_{{g}_l^n}\left(r\right)=1-\exp\left(-r\right)~\mbox{and }f_{g_l^n}\left(r\right)=\exp(-r).
\end{align}

\section{Performance Analysis}
\subsection{SCDP Analysis}
SCDP is one of the key performance metrics in wireless cache-enabled cellular networks and defined as the probability that the content requested by a typical MU is both cached in the network and successfully transmitted to the MU. Therefore, the SCDP for the considered cache-enabled mm-wave FAS is mathematically defined as
\begin{align}
P_\mathrm{scd}=\sum_{l=1}^{L}p_l\underset{P^n_{\mathrm{s},l}}{\underbrace{\Pr\left(\gamma_l\ge\eta\right)}},\label{eq-scdp-def}
\end{align}
where $\eta$ denotes the SNR threshold and $P^n_{\mathrm{s},l}$ is the probability of successful transmission.

\begin{theorem}
The SCDP for a typical MU of the considered cache-enabled mm-wave HetNet is given by \eqref{eq-scdp} (see top of the page), in which $\kappa_a$ denotes the $a$-th root of the Laguerre polynomial $\mathcal{L}_{A}\left(\kappa_a\right)$ and $A$ represents the parameter to ensure a complexity-accuracy trade-off.
\begin{figure*}[]
\begin{align}
P_\mathrm{scd}=\sum_{l=1}^L\sum_{a=1}^{A}\frac{p_l\pi q_l\mu_\mathrm{S}\mathrm{e}^{\kappa_a}\kappa_a^2\mathrm{e}^{-\pi q_l\mu_\mathrm{S}\kappa_a^2}}{\left(A+1\right)^2\mathcal{L}^2_{A+1}\left(\kappa_a\right)}\bigg[1-\Phi_\vec{R}\bigg(\sqrt 2\mathrm{erf}^{-1}\left(1-2\mathrm{e}^{-\frac{\eta\sigma^2_n\kappa_a^\alpha}{P\beta}}\right),\dots,\sqrt 2\mathrm{erf}^{-1}\left(1-2\mathrm{e}^{-\frac{\eta\sigma^2_n\kappa_a^\alpha}{P\beta}}\right);\vartheta\bigg)\bigg]\label{eq-scdp}
\end{align}
\hrulefill
\end{figure*}
\end{theorem}

\begin{proof}
From \eqref{eq-scdp-def}, $P^n_{\mathrm{s},l}$ is conditioned on the random variable $\left|X_l\right|$, which represents the distance between the typical FAS-equipped MU and its nearest serving SBS that has cached content $l$. Hence, $P^n_{\mathrm{s},l}$ can be rewritten as
\begin{align}
\hspace{-2mm}P^n_{\mathrm{s},l} &=\mathbb{E}_{\left|X_l\right|}\left\{\Pr\left(\gamma_l\ge\eta\bigg\vert \left|X_l\right|=x\right)\right\}\\
&=\int_0^\infty\Pr\left(\gamma_l\ge\eta\right)f_{\left|X_l\right|}\left(x\right)\mathrm{d}x\\
&=\int_0^\infty\Pr\left(g_l^\mathrm{FAS}  \ge\frac{\eta\sigma^2_nx^\alpha}{P\beta}\right)f_{\left|X_l\right|}\left(x\right)\mathrm{d}x, \label{eq-proof-1}
\end{align}
where $\Pr\left(g_l^\mathrm{FAS}  \ge\frac{\eta\sigma^2_nx^\alpha}{P\beta}\right)=1-F_{g^\mathrm{FAS}_l}\left(\frac{\eta\sigma^2_nx^\alpha}{P\beta}\right)$ is obtained from \eqref{eq-cdf-fas} and  $f_{\left|X_l\right|}\left(x\right)$ is the PDF of of the distance $\left|X_l\right|$, which is given by
\begin{align}
f_{\left|X_l\right|}\left(x\right)=2\pi q_l\mu_\mathrm{S}x\exp\left(-\pi q_l\mu_\mathrm{S}x^2\right).\label{eq-pdf-x}
\end{align}
Then by considering \eqref{eq-dist}, and  substituting \eqref{eq-cdf-fas} and \eqref{eq-pdf-x} into \eqref{eq-proof-1}, $P^n_{\mathrm{s},l}$ can be rewritten as
\begin{align}\nonumber
&P^n_{\mathrm{s},l}=2\pi q_l\mu_\mathrm{S}\\ \nonumber
&\times\int_0^\infty \hspace{0mm}x\mathrm{e}^{-\pi q_l\mu_\mathrm{S}x^2}\bigg[1-\Phi_\vec{R}\bigg(\sqrt 2\mathrm{erf}^{-1}\left(1-2\mathrm{e}^{-\frac{\eta\sigma^2_nx^\alpha}{P\beta}}\right),\\
&\quad\quad\dots,\sqrt 2\mathrm{erf}^{-1}\left(1-2\mathrm{e}^{-\frac{\eta\sigma^2_nx^\alpha}{P\beta}}\right);\vartheta\bigg)\bigg]\mathrm{d}x. \label{eq-proof2} 
\end{align}
Nonetheless, the integral in \eqref{eq-proof2} is mathematically intractable. As such, we utilize the Gauss-Laguerre quadrature approach, which is defined as 
\begin{align}
\int_0^\infty \Delta(s)\mathrm{d}s\approx\sum_{a=1}^{A}\frac{\mathrm{e}^{\kappa_a}\kappa_a}{2\left(A+1\right)^2\mathcal{L}^2_{A+1}\left(\kappa_a\right)}\Delta\left(\kappa_a\right),
\end{align}
in which $\kappa_a$ is the $a$-th root of Laguerre polynomial $\mathcal{L}_{A}\left(\kappa_a\right)$. By doing so, \eqref{eq-proof2} is solved after some simplifications. Then, by plugging the results into \eqref{eq-scdp-def}, the proof is completed. 
\end{proof}

\subsection{CDD Analysis}
QoS in wireless networks is intimately linked to the delay experienced by the MUs. In cache-enabled HetNets, the SBSs should deliver the requested content to the target MUs, and delay occurs primarily due to channel fading. Specifically, we consider an automatic repeat request (ARQ)-based transmission protocol such that when a MU sends its request to the nearest SBS, the requested content is repeatedly transmitted by the corresponding SBS by at most $M$ ARQ rounds (i.e., the maximum number of retransmission attempts) until it is successfully delivered. If the content is successfully delivered to the MU, then the respective SBS receives a one-bit acknowledgement message from the MU with negligible delay and error; otherwise, the SBS receives a one-bit negative acknowledgement message with the same condition. Furthermore, we assume that each ARQ round takes $T_0$ amount of time, and the outage event occurs if the MU does not successfully receive the content after $M$ ARQ rounds. The CDD for the considered system is given in the following theorem.

\begin{theorem} 
The CDD for the cache-enabled mm-wave HetNet with FAS-equipped MUs is given by \eqref{eq-cdt} (see top of next page).
\begin{figure*}[]
\begin{align}
D_\mathrm{cd}=T_0\frac{1-\left(1-\sum_{a=1}^{A}\frac{\pi q_l\mu_\mathrm{S}\mathrm{e}^{\kappa_a}\kappa_a^2\mathrm{e}^{-\pi q_l\mu_\mathrm{S}\kappa_a^2}}{\left(A+1\right)^2\mathcal{L}^2_{A+1}\left(\kappa_a\right)}\bigg[1-\Phi_\vec{R}\bigg(\sqrt 2\mathrm{erf}^{-1}\left(1-2\mathrm{e}^{-\frac{\eta\sigma^2_n\kappa_a^\alpha}{P\beta}}\right),\dots,\sqrt 2\mathrm{erf}^{-1}\left(1-2\mathrm{e}^{-\frac{\eta\sigma^2_n\kappa_a^\alpha}{P\beta}}\right);\vartheta\bigg)\bigg]\right)^M}{\sum_{a=1}^{A}\frac{\pi q_l\mu_\mathrm{S}\mathrm{e}^{\kappa_a}\kappa_a^2\mathrm{e}^{-\pi q_l\mu_\mathrm{S}\kappa_a^2}}{\left(A+1\right)^2\mathcal{L}^2_{A+1}\left(\kappa_a\right)}\bigg[1-\Phi_\vec{R}\bigg(\sqrt 2\mathrm{erf}^{-1}\left(1-2\mathrm{e}^{-\frac{\eta\sigma^2_n\kappa_a^\alpha}{P\beta}}\right),\dots,\sqrt 2\mathrm{erf}^{-1}\left(1-2\mathrm{e}^{-\frac{\eta\sigma^2_n\kappa_a^\alpha}{P\beta}}\right);\vartheta\bigg)\bigg]}\label{eq-cdt}
\end{align}
\hrulefill
\end{figure*}
\end{theorem}

\begin{proof} 
Following the same approach in \cite{chen2015backhauling}, the CDD takes at least $T_0$ amount of time with probability of $1$. Hence, the first failure happens with probability $1-P_{\mathrm{s},l}^n$ and takes $T_0$ additional time. Given the first failure, the second failure appears with probability $1-P_{\mathrm{s},l}^n$ and takes $T_0$ additional time. By proceeding in this way, the CDD can be mathematically written as
\begin{align}
D_\mathrm{c}&=T_0+T_0\left(1-P_{\mathrm{s},l}^n\right)+T_0\left(1-P_{\mathrm{s},l}^n\right)^2+\notag\\
&\quad\dots+T_0\left(1-P_{\mathrm{s},l}^n\right)^{M-1}=T_0\frac{1-\left(1-P_{\mathrm{s},l}^n\right)^M}{P_{\mathrm{s},l}^n}.\label{ref-proof2}
\end{align}
Then after computing the integral in \eqref{eq-proof2} with the help of the Gauss-Laguerre quadrature approach, and then plugging the obtained $P_{\mathrm{s},l}^n$ into \eqref{ref-proof2}, the proof is accomplished.
\end{proof}

\begin{corollary}
In the high ARQ rounds regime, i.e., $M\rightarrow\infty$, the CDD can be expressed as \eqref{eq-cdt-m} (see top of next page). 
\begin{figure*}[]
\begin{align}
\hspace{-0.4cm}D_\mathrm{cd}^\infty=T_0\left({\sum_{a=1}^{A}\frac{\pi q_l\mu_\mathrm{S}\mathrm{e}^{\kappa_a}\kappa_a^2\mathrm{e}^{-\pi q_l\mu_\mathrm{S}\kappa_a^2}}{\left(A+1\right)^2\mathcal{L}^2_{A+1}\left(\kappa_a\right)}\bigg[1-\Phi_\vec{R}\bigg(\sqrt 2\mathrm{erf}^{-1}\left(1-2\mathrm{e}^{-\frac{\eta\sigma^2_n\kappa_a^\alpha}{P\beta}}\right),\dots,\sqrt 2\mathrm{erf}^{-1}\left(1-2\mathrm{e}^{-\frac{\eta\sigma^2_n\kappa_a^\alpha}{P\beta}}\right);\vartheta\bigg)\bigg]}\right)^{-1}\label{eq-cdt-m}
\end{align}
\hrulefill\vspace{-0.45cm}
\end{figure*}
\end{corollary}

\begin{proof}
By applying $M\rightarrow\infty$ to \eqref{ref-proof2}, we have
\begin{align}
D_\mathrm{cd}^\infty=\lim_{M\rightarrow\infty} T_0\frac{1-\left(1-P_{\mathrm{s},l}^n\right)^M}{P_{\mathrm{s},l}^n}=\frac{T_0}{P_{\mathrm{s},l}^n},
\end{align}
where using the derived $P_{\mathrm{s},l}^n$, the proof is done.
\end{proof}

\begin{figure*}[!htb]
\minipage{0.32\textwidth}
\includegraphics[width=\linewidth]{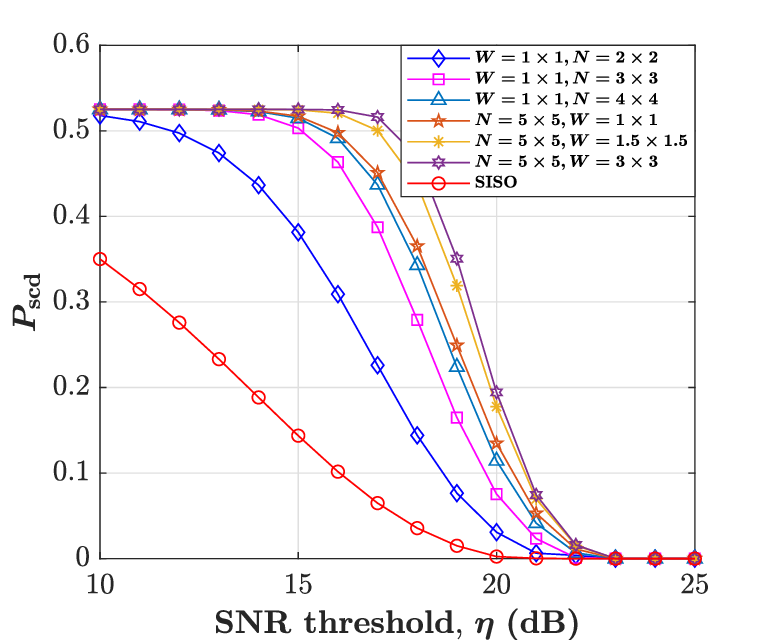}\vspace{-0.7cm}
\caption{SCDP with $q_l=1$.}\label{fig-scd-eta}
\endminipage\hfill
\minipage{0.32\textwidth}
\includegraphics[width=\linewidth]{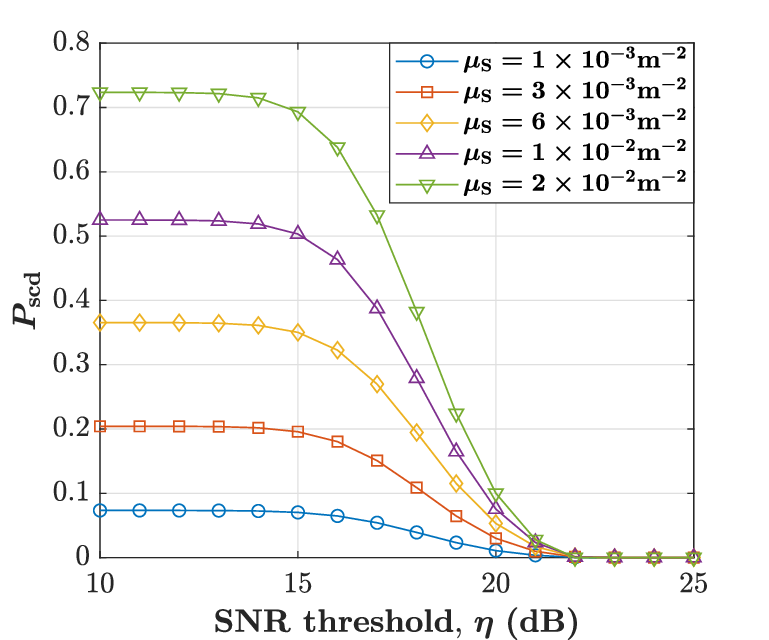}\vspace{-0.7cm}
\caption{SCDP with $q_l=1$, $N=9$, $W=1\lambda\times 1\lambda$.}\label{fig-scd-eta-mu}
\endminipage\hfill
\minipage{0.32\textwidth}%
\includegraphics[width=\linewidth]{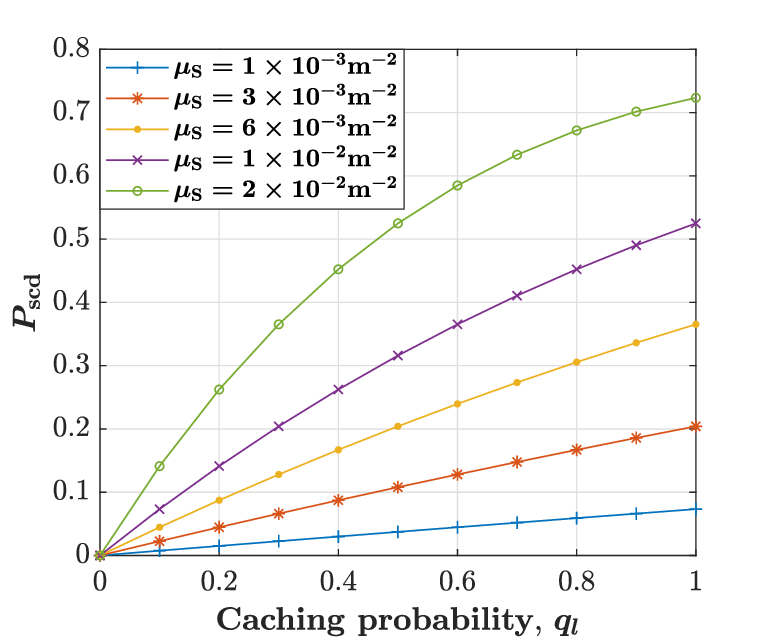}\vspace{-0.7cm}
\caption{SCDP with $N=9$, $W=1\lambda\times 1\lambda$.}\label{fig-scd-q-mu}
\endminipage \vspace{-0.2cm}
\end{figure*}
\begin{figure}[!t]
\centering
\includegraphics[width=0.7\columnwidth]{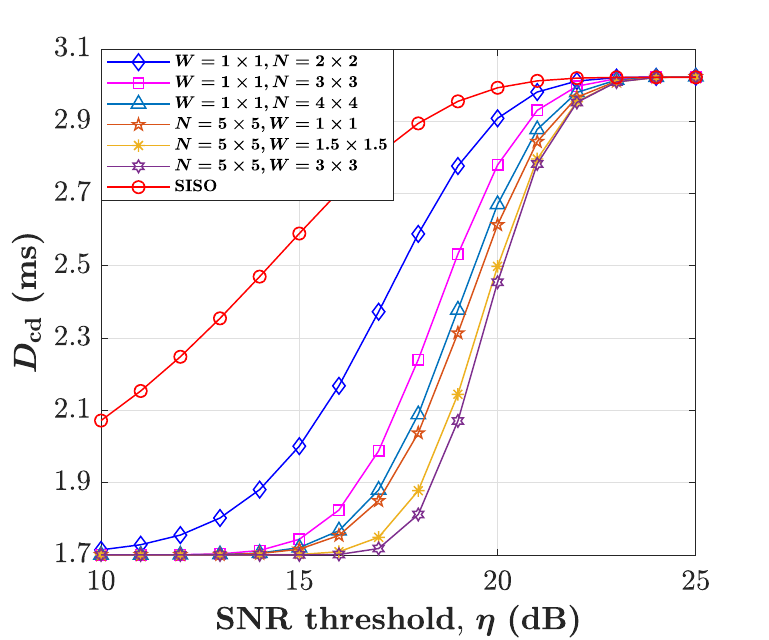}\vspace{-0.3cm}
\caption{CDD with $q_l=1$, $M=3$.}\vspace{-0.3cm}
\label{fig-cd-eta}
\end{figure}
\begin{figure}[!t]
\centering
\includegraphics[width=0.7\columnwidth]{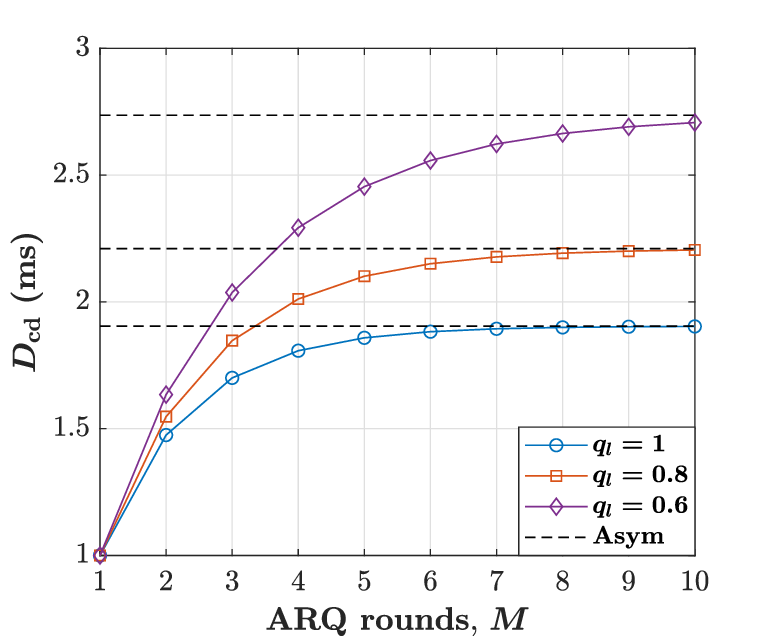}\vspace{-0.3cm}
\caption{CDD with $N=9$, $W=1\lambda\times 1\lambda$.}
\label{fig-cd-m}
\end{figure}

\section{Numerical Results}\label{sec-num}
In this section, we evaluate the considered system performance in terms of the SCDP and CDD. It is worth noting that the joint multivariate normal distribution in theoretical derivation can be estimated numerically by the mathematical package of programming languages, and the Gaussian copula can also be implemented through Algorithm 1 in \cite{ghadi2023gaussian}. Furthermore, we set the system parameters as $L=100$, $K=10$, $\mu_\mathrm{S}=10^{-2}~{\rm m}^{-2}$, $P=-30~{\rm dBm}$, $\sigma^2=-60~{\rm dBm}$, $\alpha=3$, $\beta=1$, $\eta=0~{\rm dB}$, $\zeta=1$, and $T_0=1~{\rm ms}$.

Fig.~\ref{fig-scd-eta} studies the behavior of the SCDP in terms of the SNR threshold $\eta$ for different values of fluid antenna size $W$ and various number of fluid antenna ports $N$, which corresponds to the case with caching placement probability $q_l=1$. First, we can observe that the SCDP continuously decreases as $\eta$ increases since increasing the SNR threshold means that only MUs experiencing higher SNR levels will be able to access the content and the MUs with lower SNR levels will be excluded, which thus restricts the average coverage area of the content delivery system. Next, compared with the traditional single fixed-antenna system, it can be seen that deploying FAS with only one active port at the MUs can significantly increase the chance of having reliable transmission to deliver the requested content $l$ by the MUs. Further, we can see that increasing $W$ for a constant $N$ remarkably improves the SCDP performance since it can increase the spatial separation between the fluid antenna ports, and hence the spatial correlation decreases and a higher SCDP is obtained. Besides, although increasing $N$ for a fixed $W$ increases the spatial correlation among the fluid antenna ports, it can enhance the SCDP performance since it also has the potential to boost channel capacity, diversity gain, and spatial multiplexing simultaneously.

In Fig.~\ref{fig-scd-eta-mu}, we illustrate the SCDP performance versus $\eta$ for selected values of mm-wave SBSs intensity $\mu_\mathrm{S}$. Under FAS deployment at the MUs, we can observe that increasing $\mu_\mathrm{S}$ can achieve higher values of SCDP. The main reason behind this trend is that increasing the number of mm-wave SBSs leads to a reduction in the distance between the FAS-equipped MUs and their serving SBSs; consequently, the overall SNR at the MUs is improved. Moreover, the SCDP performance against the content placement probability $q_l$ is indicated in Fig.~\ref{fig-scd-q-mu}. It can be clearly seen that the SCDP for an arbitrary content $l$ requested by the FAS-equipped MUs monotonically increases, meaning that it is a concave function of the caching placement probability. This tendency is reasonable because as $q_l$ grows, the chance of reaching the content by the MUs increases.

In Fig.~\ref{fig-cd-eta}, the CDD results in terms of the SNR threshold for different $N$ and $W$ are shown. We can observe that under a given value of the ARQ rounds $M$ and the caching placement probability $q_l$, the MUs that are equipped with the FAS receive their requested content $l$ with a lower delay compared with the fixed-antenna MUs. For instance, under a fixed $\eta=15~{\rm dB}$, the CDD for a typical FAS-equipped MU with $N=9$ and $W=1\lambda\times 1\lambda$ is around $1.72~{\rm ms}$, while the delay experienced by a typical fixed-antenna MU is almost $2.6~{\rm ms}$. Hence, under such condition, it can be seen that an arbitrary content $l$ is delivered $0.88~{\rm ms}$ faster to the MUs if the FAS is deployed. Given the importance of ARQ rounds $M$ on the QoS, Fig.~\ref{fig-cd-m} illustrates the impact of $M$ on the CDD performance. We can observe that by increasing $M$, the delay experienced by the MUs increases, which is obvious since increasing the number of attempts by the MUs to successfully receive the requested content requires additional time. We can also see that the CDD performance weakens as the caching probability $q_l$ decreases since the chance of having the requested content $l$ in the nearest mm-wave SBS to the typical MU is reduced.

\section{Conclusion}\label{sec-con}
In this letter, we analyzed the performance of cache-enabled HetNets, where the mm-wave SBSs are distributed according to an independent HPPP. We assumed that the most popular contents are cached in the SBSs to serve arbitrary content requested by the FAS-equipped MUs. In this regard, we derived the compact analytical expressions for the SCDP and CDD in terms of the joint multivariate normal distribution, using the Gauss-Laguerre quadrature technique. Eventually, numerical results showed that considering the FAS with only one active port in the MUs can improve the overall performance of cache-enabled HetNets compared with the fixed-antenna MUs.


\begin{thebibliography}{10}
	\providecommand{\url}[1]{#1}
	\csname url@samestyle\endcsname
	\providecommand{\newblock}{\relax}
	\providecommand{\bibinfo}[2]{#2}
	\providecommand{\BIBentrySTDinterwordspacing}{\spaceskip=0pt\relax}
	\providecommand{\BIBentryALTinterwordstretchfactor}{4}
	\providecommand{\BIBentryALTinterwordspacing}{\spaceskip=\fontdimen2\font plus
		\BIBentryALTinterwordstretchfactor\fontdimen3\font minus
		\fontdimen4\font\relax}
	\providecommand{\BIBforeignlanguage}[2]{{%
			\expandafter\ifx\csname l@#1\endcsname\relax
			\typeout{** WARNING: IEEEtran.bst: No hyphenation pattern has been}%
			\typeout{** loaded for the language `#1'. Using the pattern for}%
			\typeout{** the default language instead.}%
			\else
			\language=\csname l@#1\endcsname
			\fi
			#2}}
	\providecommand{\BIBdecl}{\relax}
	\BIBdecl
	
	\bibitem{CiteDrive2022}
	\BIBentryALTinterwordspacing
	I.~ERICSSON. (2022) \textit{ERICSSON Mobility Report}. [Online]. Available:
	\url{https://www.ericsson.com/en/reports-and-papers/mobility-report}
	\BIBentrySTDinterwordspacing
	
	
\bibitem{bastug2014living}
E.~Bastug, M.~Bennis, and M.~Debbah, ``{Living on the edge: The role of	 proactive caching in 5G wireless networks},'' \emph{IEEE Commun. Mag.}, vol.~52, no.~8, pp. 82--89, Aug. 2014.
	
\bibitem{wong2020fluid}
K.-K. Wong, A.~Shojaeifard, K.-F. Tong, and Y.~Zhang, ``{Fluid antenna systems},'' \emph{IEEE Trans. Wirel. Commun.}, vol.~20, no.~3, pp. 1950--1962, Mar. 2021.

%

\bibitem{blaszczyszyn2015optimal}
B.~Blaszczyszyn and A.~Giovanidis, ``{Optimal geographic caching in cellular networks},'' in \emph{Proc. IEEE Inter. Conf. Commun.}, pp. 3358--3363, 8-12 Jun. 2015, London, United Kingdom.

\bibitem{wang2017cooperative}
Y.~Wang, X.~Tao, X.~Zhang, and Y.~Gu, ``{Cooperative caching placement in cache-enabled D2D underlaid cellular network},'' \emph{IEEE Commun. Lett.}, vol.~21, no.~5, pp. 1151--1154, May 2017.
	
\bibitem{zhu2018content}
Y.~Zhu, G.~Zheng, L.~Wang, K.-K. Wong, and L.~Zhao, ``{Content placement in cache-enabled sub-6 GHz and millimeter-wave multi-antenna dense small cell networks},'' \emph{IEEE Trans. Wirel. Commun.}, vol.~17, no.~5, pp. 2843--2856, May 2018.
	
\bibitem{bacstug2015cache}
E.~Ba{\c{s}}tu{\u{g}}, M.~Bennis, M.~Kountouris, and M.~Debbah, ``{Cache-enabled small cell networks: Modeling and tradeoffs},'' \emph{EURASIP J. Wirel. Commun. Netw.}, no. 41, pp. 1--11, Feb. 2015.
	
	
\bibitem{rostami2019outage}
F.~Rostami~Ghadi and M.~R. Javan, ``{Outage and delay performance of content caching in two-tier cooperative cellular networks},'' \emph{IET Commun.}, vol.~13, no.~16, pp. 2492--2499, Oct. 2019.

\bibitem{wang2021content}
W.~Wang \textit{et al.},  ``{Content delivery analysis in cellular networks with aerial caching and mmwave backhaul},'' \emph{IEEE Trans. Veh. Technol.}, vol.~70, no.~5, pp. 4809--4822, May 2021.

\bibitem{wong2021fluid}
K.-K. Wong and K.-F. Tong, ``{Fluid antenna multiple access},'' \emph{IEEE Trans. Wirel. Commun.}, vol.~21, no.~7, pp. 4801--4815, Jul. 2022.
\bibitem{khammassi2023new}
M.~Khammassi, A.~Kammoun, and M.-S. Alouini, ``{A new analytical approximation of the fluid antenna system channel},'' \emph{IEEE Trans. Wirel. Commun.}, vol. 22, no. 12, pp. 8843--8858, Dec. 2023.
\bibitem{ghadi2023gaussian}
F.~R. Ghadi \textit{et al.}, ``A Gaussian copula approach to the performance analysis of fluid antenna systems,'' \emph{arXiv preprint}, \url{arXiv:2309.07506}, 2023.

\bibitem{skouroumounis2022fluid}
C.~Skouroumounis and I.~Krikidis, ``{Fluid antenna with linear MMSE channel estimation for large-scale cellular networks},'' \emph{IEEE Trans. Commun.}, vol.~71, no.~2, pp. 1112--1125, Feb. 2023.
\bibitem{xu2023channel}
H.~Xu \textit{et al.},  ``{Channel estimation for FAS-assisted multiuser mmWave systems},'' \emph{IEEE Commun. Lett.}, vol. 28, no. 3, pp. 632--636, Mar. 2024.
\bibitem{zhang2023successive}
Z.~Zhang, J.~Zhu, L.~Dai, and R.~W. Heath~Jr, ``Successive Bayesian reconstructor for channel estimation in fluid antenna systems,'' {\em arXiv preprint}, \url{arXiv:2312.06551v3}, Jan. 2024.
	
	
\bibitem{tlebaldiyeva2022enhancing}
L.~Tlebaldiyeva, G.~Nauryzbayev, S.~Arzykulov, A.~Eltawil, and T.~Tsiftsis, ``Enhancing {QoS} through fluid antenna systems over correlated {Nakagami}-$m$ fading channels,'' in {\em Proc. IEEE Wireless Commun. Netw. Conf.}, pp.~78--83, 10-13 Apr. 2022, Austin, TX, USA.
\bibitem{new2023fluid}
W.~K. New, K.-K. Wong, H.~Xu, K.-F. Tong, and C.-B. Chae, ``Fluid antenna system: New insights on outage probability and diversity gain,'' {\em IEEE Trans. Wireless Commun.}, vol. 23, no. 1, pp. 128--140, Jan. 2024.
\bibitem{ghadi2023copula}
F.~R. Ghadi, K.-K. Wong, F.~J. L{\'o}pez-Mart{\'\i}nez, and K.-F. Tong, ``Copula-based performance analysis for fluid antenna systems under arbitrary fading channels,'' {\em IEEE Commun. Lett.}, vol.~27, no.~11, pp.~3068--3072, Nov. 2023.
\bibitem{vega2023novel}
J.~D. Vega-S{\'a}nchez, A.~E. L{\'o}pez-Ram{\'\i}rez, L.~Urquiza-Aguiar, and D.~P.~M. Osorio, ``{Novel expressions for the outage probability and diversity gains in fluid antenna system},'' \emph{IEEE Wirel. Commun. Lett.}, vol. 13, no. 2, pp. 372--376, Feb. 2024
\bibitem{ghadi2023fluid}
F.~R. Ghadi {\em et al.}, ``Fluid antenna-assisted dirty multiple access channels over composite fading,'' {\em IEEE Commun. Lett.}, vol. 28, no. 2, pp. 382--386, Feb. 2024.
\bibitem{ghadi2024fluid}
F.~R. Ghadi {\em et al.}, ``On Performance of RIS-Aided Fluid Antenna Systems,'' {\em arXiv preprint}, \url{arXiv:2402.16116}, Feb. 2024.

\bibitem{10354059}
Y.~Chen, S.~Li, Y.~Hou, and X.~Tao, ``Energy-efficiency optimization for slow fluid antenna multiple access using mean-field game,'' {\em IEEE Wireless Commun. Lett.}, early access, \url{doi:10.1109/LWC.2023.3341460}, 2023.
\bibitem{ye2023fluid}
Y.~Ye {\em et al.}, ``Fluid antenna-assisted {MIMO} transmission exploiting statistical {CSI},'' {\em IEEE Commun. Lett.}, vol.~28, no.~1, pp.~223--227, Jan. 2024.
\bibitem{new2023fluid1}
W.~K. New {\em et al.}, ``Fluid antenna system enhancing orthogonal and non-orthogonal multiple access,'' {\em IEEE Commun. Lett.}, vol. 28, no. 1, pp. 218--222, Jan. 2024.
\bibitem{ISAC_FAS}
C.~Wang {\em et al.}, ``Fluid antenna system liberating multiuser {MIMO} for {ISAC} via deep reinforcement learning,'' accepted in {\em IEEE Trans. Wireless Commun.}, 2024.

\bibitem{zhu2016wireless}
Y.~Zhu, L.~Wang, K.-K. Wong, S.~Jin, and Z.~Zheng, ``{Wireless power transfer in massive MIMO-aided HetNets with user association},'' \emph{IEEE Trans. Commun.}, vol.~64, no.~10, pp. 4181--4195, Oct. 2016.

\bibitem{singh2015tractable}
S.~Singh, M.~N. Kulkarni, A.~Ghosh, and J.~G. Andrews, ``{Tractable model for rate in self-backhauled millimeter wave cellular networks},'' \emph{IEEE J. Sel. Areas Commun.}, vol.~33, no.~10, pp. 2196--2211, Oct. 2015.

\bibitem{andrews2016modeling}
J.~G. Andrews \textit{et al.},  ``{Modeling and analyzing millimeter wave cellular systems},'' \emph{IEEE Trans. Commun.}, vol.~65, no.~1, pp. 403--430, Jan. 2017.

\bibitem{chen2015backhauling}
D.~C. Chen, T.~Q. Quek, and M.~Kountouris, ``{Backhauling in heterogeneous cellular networks: Modeling and tradeoffs},'' \emph{IEEE Trans. Wirel. Commun.}, vol.~14, no.~6, pp. 3194--3206, Jun. 2015.	
\end{thebibliography}


\end{document}